\documentclass[conference, a4paper]{IEEEtran}
\usepackage{amsmath,amsfonts}
\usepackage{array}
\usepackage{amsfonts,amssymb}
\usepackage[noend]{algpseudocode}
\usepackage{algorithm}
\usepackage{algorithmicx}
\usepackage{algpseudocode}
\usepackage{bm}
\usepackage{balance}
\usepackage{booktabs}
\usepackage{changepage}
\usepackage{cite}
\usepackage{color}
\usepackage{diagbox}
\usepackage{enumitem}
\usepackage{footnote}
\usepackage{float}
\usepackage[T1]{fontenc}
\usepackage{graphicx}
\usepackage[utf8]{inputenc}
\usepackage{indentfirst}
\usepackage{mathrsfs}
\usepackage{multirow} 
\usepackage{makecell}
\usepackage{pifont}
\usepackage{subfigure}
\usepackage{setspace}
\usepackage{stfloats}
\usepackage{textcomp}
\usepackage{times}
\usepackage{threeparttable}
\usepackage{url}
\usepackage{verbatim}
\usepackage{xcolor}

\def\BibTeX{{\rm B\kern-.05em{\sc i\kern-.025em b}\kern-.08em
		T\kern-.1667em\lower.7ex\hbox{E}\kern-.125emX}}

\newtheorem{myDef}{\textbf{Definition}}
\newtheorem{lemma}{\textbf{Lemma}}
\newtheorem{prop}{\textbf{Proposition}}
\newenvironment{proof}{{\indent  \it Proof:}}

\begin{document}
	\title{Duration-adaptive Video Highlight Pre-caching for Vehicular Communication Network}
	\author{
		
		\IEEEauthorblockN{Liang Xu\IEEEauthorrefmark{1}, Deshi Li\IEEEauthorrefmark{1}, Kaitao Meng\IEEEauthorrefmark{2}, Mingliu Liu\IEEEauthorrefmark{3} and Shuya Zhu\IEEEauthorrefmark{1}}
		
		\IEEEauthorblockA{\IEEEauthorrefmark{1}Electronic Information School, Wuhan University, Wuhan, China.}
		
		\IEEEauthorblockA{\IEEEauthorrefmark{2}Department of Electronic and Electrical Engineering, University College London, London, UK}
		
		\IEEEauthorblockA{\IEEEauthorrefmark{3}State Grid Hubei Electric Power Research Institute, Wuhan, China}
		
		
		Emails: \IEEEauthorrefmark{1}\{lgxu, dsli, shuyazhu\}@whu.edu.cn, 
		\IEEEauthorrefmark{2}\{kaitao.meng\}@ucl.ac.uk,
		\IEEEauthorrefmark{3}\{liumingliu\}@whu.edu.cn
	}

	\maketitle
	
	
	\vspace{-10mm}
	\begin{abstract}
		Video traffic in vehicular communication networks (VCNs) faces exponential growth. However, different segments of most videos reveal various attractiveness for viewers, and the pre-caching decision is greatly affected by the dynamic service duration that edge nodes can provide services for mobile vehicles driving along a road. In this paper, we propose an efficient video highlight pre-caching scheme in the vehicular communication network, adapting to the service duration. Specifically, a highlight entropy model is devised with the consideration of the segments' popularity and continuity between segments within a period of time, based on which, an optimization problem of video highlight pre-caching is formulated. As this problem is non-convex and lacks a closed-form expression of the objective function, we decouple multiple variables by deriving candidate highlight segmentations of videos through wavelet transform, which can significantly reduce the complexity of highlight pre-caching. Then the problem is solved iteratively by a highlight-direction trimming algorithm, which is proven to be locally optimal. Simulation results based on real-world video datasets demonstrate significant improvement in highlight entropy and jitter compared to benchmark schemes.
	\end{abstract}
	
	
	\section{Introduction}
	The emerging technology of the 6G wireless system has been introduced to intelligent	transportation, which would accelerate the development of autonomous driving \cite{bDT}, and the promising application of autonomous vehicles gives rise to the prosperity of content services \cite{b40}. In particular, video services, such as short video, video-on-demand (VoD), etc., will account for an extremely high ratio of future content traffic \cite{b10}. Moreover, driven by the new video applications, such as augmented reality (AR) navigation, virtual reality (VR), and 4K/8K video delivery, video traffic in vehicular communication networks (VCNs) faces exponential growth \cite{b46}. In this regard, mobile edge caching is proposed to sink storage resources to close vehicular nodes and cache videos required by users to edge nodes (e.g., roadside units (RSUs)), which can significantly reduce the network backhaul traffic and latency. However, due to the limited cache capacity of RSUs and the high mobility of vehicles, video pre-caching is still a challenging problem in vehicular communication networks \cite{b1}.\par
	
	In the literature, there are some related works devoted to proactive caching in RSUs in vehicular communication networks. To make intelligent prefetching decisions, a multi-tier caching mechanism assisted by vehicle mobility prediction was proposed in \cite{b34}. Furthermore, to ensure in-order delivery of video chunks adapted to the mobility characteristics of connected cars, the authors in \cite{b41} proposed a roadside prefetching in RSUs that optimally caches content chunks required most at edge nodes. However, when a user on road requests a video service, the service duration available for the user to watch the video varies, depending on the distance and speed between their departure and termination locations. Most of the time the service duration is too short to finish watching a movie. Furthermore, users often skip certain segments while watching some videos \cite{bperchunkcaching}, and pre-caching whole-file may consume extra storage resources. The statistical analysis conducted in \cite{chunkcharacterize} has revealed that, on average, only $60$ percent of each video file is actually watched. Hence, conventional whole-file pre-caching strategies will result in a	huge waste of resources. It is worth noting that the caching performance is greatly affected by user behaviors of browsing videos/documents, i.e., skipping, switching, dragging, etc., which is seldom considered in the literature. We provide a summary of user browsing behaviors and the corresponding video segment pre-caching strategies when watching different types of videos in Table \ref{Table1}. Indeed, each video segment varies with the popularity that characterizes its attractiveness, and users tend to jump directly to segments with higher popularity by dragging or dropping the progress bar. Additionally, these popular segments often receive a higher number of bullet-screen comments from users. Hence, investigating a more realistic video pre-caching according to video segments' popularity can not only reduce resources but also improve cache efficiency.\par
	
	In general, the popularity value of video segments exhibits a positive association with metrics such as playout times, the number of bullet-screen comments, and ratings provided by the public or professionals \cite{meng2022sensing}. Recently, there has been an increase in the generation and accessibility of popularity value for video segments by various content providers and web pages. For instance, the popularity value of YouTube segments can be inferred from their replay frequency, as exemplified by the utilization of the video activity graph employed by YouTube \cite{VCG}. However, only caching the popular segment may result in great viewing jitter when there is a significant incoherent plot between two watching segments, as it ignores another important performance indicator, i.e., the continuity between adjacent segments. Due to the limited storage and communication resources of RSUs, it is hard to balance the popularity of segments and the continuity between segments within a period of time. Furthermore, the pre-caching performance is greatly affected by the dynamic service duration of vehicles on road, especially for dynamic driving vehicles. Therefore, there is an urgent need to investigate video segment pre-caching that takes into account both video segments' popularity and continuity, adapting to the service duration on road, which means the composition of each segment is dynamically determined according to the user's request or the service duration.\par
	\begin{table*}[t] 
		\centering
		\vspace{-4mm}
		\caption{Browsing Behaviors of Users Watching Different Types of Videos.}
		\vspace{-3mm}
		\label{Table1}
		\begin{tabular}{ | c | c | c | c | c |}
			\hline
			\bf{Type} & \bf{Applications} & \bf{\makecell[c]{Scanning/Browsing behaviors}} & \bf{\makecell[c]{Precached segments}} &
			\bf{\makecell[c]{Pre-caching strategies}}\\ 
			\hline
			\multirow{2}{*}{\makecell[c]{\\ VoD}} & \multirow{2}{*}{\makecell[c]{ Mi TV, IQiyi, \\ YuTube, \\ Netflix}} & \makecell[c]{Fasting forward at N times speed} & \makecell[c]{Chunk after a certain interval} & \makecell[c]{N-Speed playback} \\ \cline{3-5} 
			&{} &\makecell[c]{Skipping, playbacking, \\ sending bullet-screen comments } & \makecell[c]{Captivating segments} & \makecell[c]{Selecting segments with po\\-pularity as high as possible} \\ \hline
			\multirow{2}{*}{\makecell[c]{Short \\ video} } &  \multirow{2}{*}{\makecell[c]{Tik Tok}} &
			\multirow{2}{*}{\makecell[c]{Switching, dragging, dropping, \\ sending bullet-screen comments }} & \multirow{2}{*}{\makecell[c]{Captivating segments}} & \multirow{2}{*}{\makecell[c]{Selecting segments with po\\-pularity as high as possible}} \\
			&{} &\makecell[c]{} & \makecell[c]{} & \makecell[c]{}\\
			\hline
		\end{tabular}
		\vspace{-6mm}
	\end{table*} 
	
	To provide efficient video services for vehicular communication networks, we develop a Duration-adaptive Highlight Pre-Caching (DHPC) scheme. Leveraging the popularity of video segments, we design a highlight segments pre-caching method, which dynamically adapts to the varying service duration. Specifically, a video viewing quality evaluation model based on video highlight entropy is devised to integrate popularity and continuity together, based on which, the average highlight entropy of requested videos is maximized. To solve the optimization problem, we first propose a highlight segmentation method by wavelet transform to decouple the optimized variables, then propose a highlight-direction trimming algorithm to compose optimal pre-caching video files in RSUs efficiently. The main contributions of this paper are summarized as follows:
	\begin{itemize}[leftmargin=*]
		\item A duration-adaptive highlight pre-caching scheme according to video segments' popularity is proposed to provide efficient video services for vehicular communication networks, which adapts to dynamic service durations of driving vehicles.
		\item To balance the popularity of segments and continuity between segments, a highlight entropy based quality evaluation model is constructed. Furthermore, we formulate the optimization problem of highlight entropy maximization while improving the quality of the viewing experience.
		\item To solve the problem and reduce the computational complexity, wavelet transformation is conducted to obtain the candidate highlight segments. Then, a highlight-direction trimming algorithm is proposed to quickly obtain optimal pre-caching video files in RSUs.
	\end{itemize}
	The remainder of this article is organized as follows. Section II presents the system model and video viewing quality evaluation model. The DHPC scheme is illustrated in Section III. In Section IV, the proposed methods are evaluated through abundant simulations, and finally, we conclude this article and direct our future work to Section V. 
	
	\section{System Model}
	\label{SYSTEM}
	As illustrated in Fig.~{\ref{fig2}}, considering a duration-adaptive highlight pre-caching framework for the vehicular communication network, there are various cloud video content servers, a macro base station (MBS), and $M$ RSUs deployed along road. The set of the RSUs and that of the allocated storage sizes are denoted by $\mathcal{M} = \left\lbrace 1, \cdots, M\right\rbrace $ and $\mathcal{C} = \left\lbrace c_{1}, \cdots, c_{M}\right\rbrace $, respectively. In addition to video caching, the RSU is also responsible for video transmission for users in a vehicle. The downlink transmission rate from RSU $m$ to the user can be expressed as $r_m$. The cloud video servers predict video segments’ popularity of each video based on the information on historical user browsing behaviors. The MBS, equipped with caching, computing, and communication capabilities, serves as a central controller to manage the resources of RSUs by determining how to segment and trim highlights pre-cached to RSUs according to the dynamic service duration.
	\par
	
	\subsection{Video Chunk Popularity Model}
	In the vehicular communication network, the video requested by a user in the vehicle is denoted by $f\in\mathcal{F}=\left\lbrace 1, \cdots, F\right\rbrace $, and the popularity is assumed to obey the Zifp distribution in \cite{b25}. Then the popularity of video $f$ can be expressed as $p_{f}=\left( {1/f^{\beta}}\right) /\left( \sum_{i=1}^{F}1/i^{\beta}\right) $, and $\beta$ is the Zipf exponent. Since segments' composition is variable, to analyze the popularity of segments, the basic unit of segment composition needs to be modeled. The definition of the basic segment composition unit is given in the following.
	\begin{myDef}
		Chunk: A chunk is the smallest divisible video unit watched by users, and contains ${\tau}$ continuous video frames.
	\end{myDef}

	Then, the whole video $f$ can be equally divided into $X_f$ ordered chunks, and the chunk in video $f$ is indexed by $x\in{\cal{X}}_{f}=\{1, \cdots, X_f\}$. The value of video chunks' popularity can be expressed as the ratio of the number of chunks watched or the ratio of the number of bullet-screen comments obtained by users' historical browsing datasets. The video chunk popularity of video $f$ is a time series and is denoted by $y_f(x)$, where $y_f(x) \in \left[ 0, 1 \right], x\in{\cal{X}}_{f}$. A video segment is formed by a sequence of contiguous chunks with the same caching state, indexed by $k \in {\mathcal{K}}_{f}=\left\lbrace 1, \cdots, K_f \right\rbrace$. Video segments' popularity value is the sum of these chunks. \par
	\subsection{Service Duration Model}
	Suppose the vehicle enters RSU $1$ to initiate video service, and exits RSU $M$ to conclude video service, the service duration of the vehicle on road can be represented by the sum of the dwell times in these $M$ RSUs. To model service duration, the dwell time under each RSU needs to be modeled first. The dwell time of the vehicle under RSU $m$ is denoted by $t_m = {d_m}/{v_m}$, where $d_m \in \mathcal{D}$ and $v_{m} \in \mathcal{V}$ denotes the coverage ranges and the average vehicle's speed of RSU the $m$. Then, the service duration of the vehicle on road can be expressed as $T_{d} = \sum_{m=1}^{M}t_{m}$. The transmission traffic of video chunks in each RSU and the number of playout video chunks during the service duration can be given by $Z_{r} = \sum_{m=1}^{M} r_m t_m$ and $Z_p = T_{d}z_t/\tau$, where $z_t$ denotes the number of video chunks within one second.
	\par
	\begin{figure}[t]
		\centering
		\vspace{-0.4cm}
		\includegraphics[width=8.5cm]{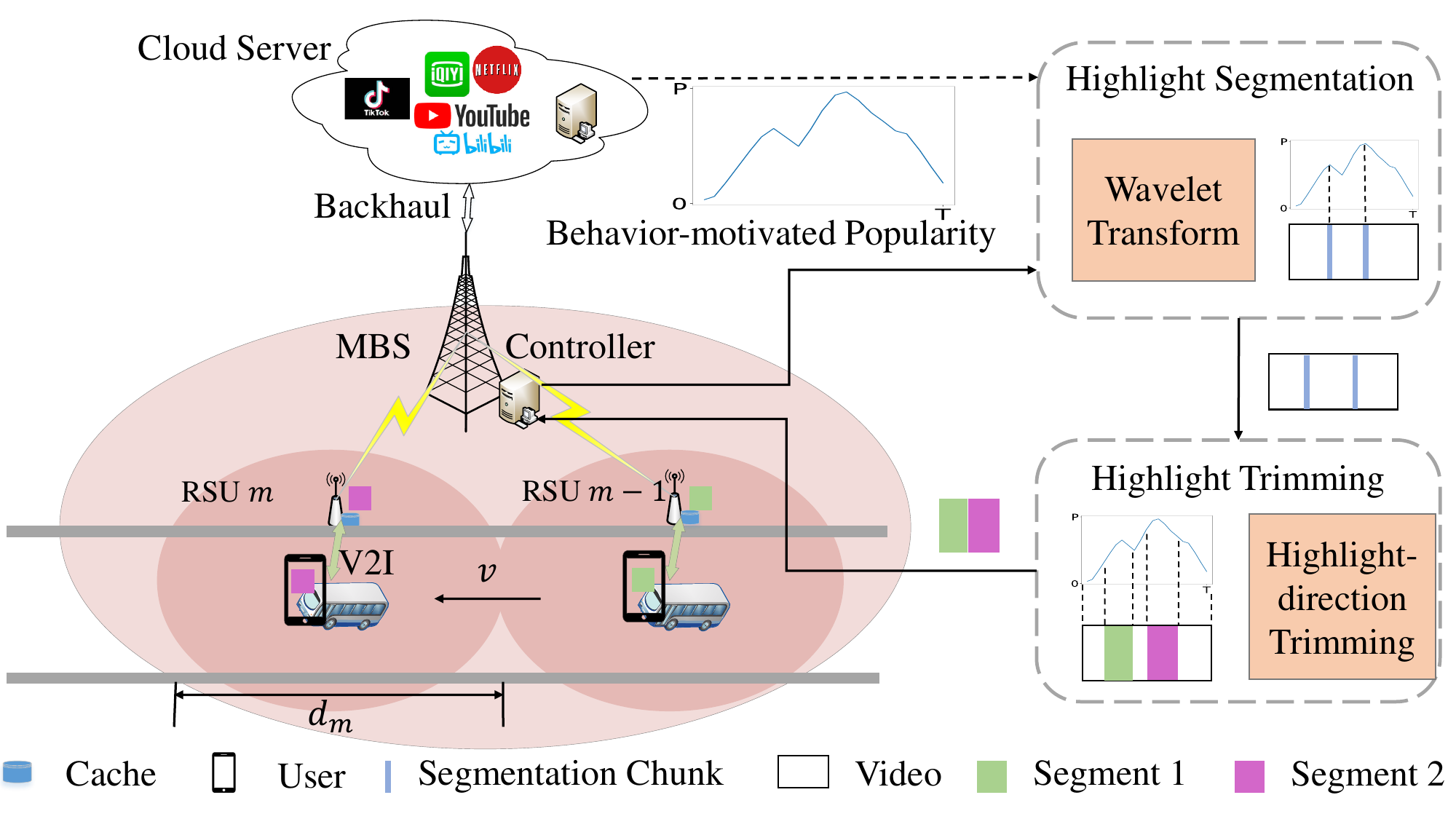}
		\vspace{-0.3cm}
		\caption{The illustration of video pre-caching in RSUs scenarios.}
		\vspace{-0.6cm}
		\label{fig2}
	\end{figure}
	\subsection{Wavelet Transform for Video Chunk Popularity}
	Due to the noise in users' historical browsing datasets, the method of directly sorting the values of video chunks’ popularity may involve many peak values of non-highlight segments, resulting in inaccurate segmentation of video highlight segments. Since varying frequencies of video chunks’ popularity time series in each segment, conventional frequency-domain analysis methods may not be able to accurately locate the chunk where the peak value appears. The continuous wavelet transform (CWT) is the time-frequency analysis method, with the frequency of the video chunks’ popularity value changes being observed step by step from coarse to fine, which can locate the candidate highlight based on the video chunks’ popularity \cite{b26}. As a result, the candidate highlights segmentation can be obtained. The CWT for video chunk popularity of video $f$  is as
	\vspace{-4mm}
	\begin{equation}
		W_{f}(b, s)=\dfrac{1}{\sqrt{|s|}}\int_{-\infty}^{\infty} y_{f}(x)\psi^{\textasteriskcentered}(\dfrac{x-b}{s})dx \label{WT}, 
	\end{equation}
	where $\psi(x)$ is called the wavelet basis function, $s$ and $b$ are scale and shift of the wavelet basis function along the $x$ domain. \par
	\subsection{Highlight Entropy Model}
	To facilitate effective video pre-caching in the scenario where users skip some segments, popularity expectation of users' viewing segments with different pre-caching chunks should be evaluated first. Based on video chunks’ popularity, we define the following evaluation model for highlight pre-caching strategy.
	\begin{myDef}
		\textit{Highlight Entropy (HE).} The average video segments’ popularity expectation for viewing segments after each skipped segment when caching chunks of highlight segments at RSUs.
	\end{myDef}
		The highlight entropy of viewing video $f$ is expressed as
		\begin{equation}\label{HE}	
			E_f  =\dfrac{p_f}{N_f}\sqrt{\dfrac{\sum_{k=1}^{K_f}\left(\theta_{x,k}^s\sum_{x=1}^{X_f}\theta_{x,f}^cy_{f}\left( x\right)\right) ^2}{\sum_{k=1}^{K_f}\sum_{x=1}^{X_f}\left( \left(1- \theta_{x,k}^s\theta_{x,f}^c\right) y_{f}(x)\right)^2} },
		\end{equation}
		where the numerator underneath the square root symbol represents the summation of the squares of the video chunks’ popularity of cached video highlights, and the denominator in the root sign denotes the sum of the squares of the video chunks’ popularity of uncached video chunks. $N_f = \sum_{k=1}^{K_f}\sum_{x=1}^{X_f}\left| \theta_{x,k}^s\theta_{x,f}^c -\theta_{x-1,k-1}^s\theta_{x-1,f}^c \right| $ indicates the number of skipped segments of video $f$. $\theta _{x,f}^c$ is a binary variable that represents the chunk caching decision for video $f$, and $\theta _{x,f}^c = 1$ if chunk $x$ is cached, otherwise, $\theta _{x,f}^c = 0$. The segmentation status of video chunk $x$ can be represented by a binary variable $\theta _{x,k}^s$, where $\theta _{x,k}^s = 1$ means chunk $x$ is in segment $k$. When viewing video segments that encompass a greater number of chunks with elevated popularity, the highlight entropy of the video is higher. \par
		\subsection{Problem Formulation}
		This work aims to maximize the highlight entropy of requested videos by optimizing video files pre-caching, subject to transmission resource, storage resource, and service duration. Accordingly, the optimization problem can be formulated as
		\begin{alignat}{2}
			\vspace{-1mm}
			\label{P0}
			(\rm{P0}): \quad & \begin{array}{*{20}{c}}
				\mathop {\max }\limits_{\theta_{x,f}^c\in\Theta^c, \theta_{x,k}^s\in\Theta^s} \sum_{f=1}^{F}E_f \end{array} & \\ 
			\mbox{s.t.}\quad
			& \sum_{f=1}^{F}\sum_{x=1}^{X_f}\!\theta_{x,f}^c z_f\!\leq \! \min\!\left\lbrace\! Z_r\!, \!\sum_{m=1}^{M} c_{m},\! \dfrac{Z_pF}{\sum_{f=1}^{F}z_f}\!\right\rbrace\!, & \tag{\ref{P0}a}\\
			& \sum_{f=1}^{F}\theta_{x,f}^c \leq 1, \forall x\in{\mathcal{X}}_f, & \tag{\ref{P0}b} \\
			& \theta_{x,f}^c, \theta_{x,k}^s \in \left\lbrace 0,1 \right\rbrace, \forall f\in\mathcal{F}, k\in{\mathcal{K}}_f, x\in{\mathcal{X}}_f, & \tag{\ref{P0}c} \\
			& \sum_{k=1}^{K_f}\theta_{x,k}^s = \theta_{x,f}^c, \forall f \in\mathcal{F}, \forall x\in{\mathcal{X}}_f, & \tag{\ref{P0}d} \\
			& \!\theta_{x,f}^c\! - \!\theta_{x,k}^s \!= \!\theta_{x',f}^c \! -\!\theta_{x',k}^s, \!\forall f\!\in\mathcal{F}, \! k\!\in{\mathcal{K}}_f, \!\left\lbrace \! x'\!, \! x \!\right\rbrace \!\in{\! \mathcal{X}}_f. \! & \tag{\ref{P0}e} 
		\end{alignat}
		In problem (P0), constraint (\ref{P0}a) means that the cached chunk size is limited by service duration, network bandwidth, and the cache storage capacity, where $z_{f}$ represents the chunk size of video $f$. Constraint (\ref{P0}b), (\ref{P0}c), and (\ref{P0}d) denote the binary pre-caching decision constraints, constraint (\ref{P0}e) implies that if video chunk $x$ and $x'$ are in segment $k$, then the corresponding chunk cache status must be the same. 
		\par
		Solving (P0) is challenging due to the following reasons. First, it is non-convex and lacks a closed-form objective function. Second, the chunk caching decision will impact the result of highlight segmentation and trimming, further greatly complicating the problem-solving. Third, the RSU dwelling time is too short, generally in order of minutes, which requires the algorithm to operate efficiently. Therefore, a highlight-direction local optimal algorithm is proposed to reduce the computation	complexity caused by blind segmentation diversity, which makes timely highlight segmentation and trimming decisions as a vehicle travels along the road.\par
		\section{Algorithm Design and Analysis}
		In this section, the highlight entropy maximization problem (P0) by the candidate highlights segmentation is reformed. Then, a highlight-direction trimming algorithm is proposed. Since this problem couples the chunk caching decisions and the segment selections of videos complicating the problem-solving, the wavelet transform is introduced to obtain the initial segmentation, thereby speeding up algorithm convergence.\par
		\subsection{Problem Transformation and Decomposition}
		The optimal solution to (P0) is strongly influenced by video chunks’ popularity, resulting in cached video highlights containing chunks with larger video chunks’ popularity values. In this subsection, video candidate highlight segmentations are obtained by the CWT for video chunk popularity and transform the (P0) into a step-by-step solvable form and derive the optimal direction for the subsequent step.\par
		It can be found that, if $\sum_{x=1}^{X_f}\theta_{x,f}^c=1$, for any given video $f$, the optimal highlight segment $x^*=\arg \mathop {\max }\limits_{x\in{\mathcal{X}}_f } y_f(x)$. Otherwise, if $\sum_{x=1}^{X_f}\theta_{x,f}^c>1$, the optimal highlight segments are highly coupled with the segments' continuity. Although optimal highlight segments do not contain all the chunks with peak values of video chunks’ popularity, these chunks are able to derive the candidate highlights segmentation for the optimization problem (P0).\par
		\begin{lemma}\label{MM}
			When $\sum_{x=1}^{X_f}\theta_{x,f}^c>1$ and $\sum_{x=a}^{b}\theta_{x,k}^s \geq 1, x\in\left[ a, b\right] $, for any given video $f$, $\exists$ $ \theta_{x_m,f}^c=1$, where $x_m = \arg \mathop {\max }\limits_{x\in\left[ a, b \right] } y_f(x)$.
		\end{lemma}
		
		\begin{proof}
			Taking a chunk $x_0$ of video $f$ in segment $k$ starting from chunk $a$ and ending with chunk $b$, and where $ x_{m} \neq x_0 \in \left\lbrace a, b\right\rbrace $, then the number of skipped chunks $N_f = X_f-(b-a+1) $ for video $f$ When $\sum_{x=1}^{X_f}\theta_{x,f}^c>1$ and $\sum_{x=a}^{b}\theta_{x,k}^s \geq 1, x\in\left[ a, b\right] $. The sum of video chunks’ popularity for pre-cached segment $k$ square yields
			\begin{equation}
				\sum_{x=a}^{b}y_f(x)-y_f(x_0))^2 \geq (\sum_{x=a}^{b}y_f(x)-y_f(x_m))^2, 
			\end{equation}
			which holds for all $x_{m} \neq x_0 \in \left\lbrace a, b\right\rbrace $. And the sum of skipped video chunks' popularity square satisfies
			\begin{equation}
				\begin{split}
					(y_f(x_0))^2
					\leq (y_f(x_m))^2, 
				\end{split}
			\end{equation}
			which holds for all $x_{m} \neq x_0 \in \left\lbrace a, b\right\rbrace$. Then the highlight entropy obtained by caching chunk $x_{m}$ yields
			\begin{equation}
				\begin{aligned}
					& E_f\left( x_{m}\right) - E_f\left( x_0\right)\\
					&	= \dfrac{p_f}{N_f} \sqrt{I_{f}(x_m)}
					-			\dfrac{p_f}{N_f}\sqrt{I_{f}(x_0)} \geq 0, 
				\end{aligned}
			\end{equation}
			holds for all $k \!\in \!{\mathcal{K}}_f, f \in \mathcal{F}$, where function $I_{f}(x)$ is expressed as $I_{f}(x)= (\sum_{x'=a}^{b}y_f(x')-y_f(x))^2/(\sum_{x'=0}^{a-1}(y_f(x'))^2+\sum_{x'=b+1}^{X_f}(y_f(x'))^2+(y_f(x))^2)$. To maximize the highlight entropy, the chunk with the maximal peak value of video chunks' popularity will be pre-cached. $\hfill\blacksquare$
		\end{proof}
		
		According to Lemma \ref{MM}, chunks with peak values of video chunks' popularity are key initial sets for candidate highlight segmentations. It has been stated in \cite{b26} that the wavelet modulus maxima of a real wavelet are capable of identifying all singular points within a specified interval. As the scale decreases, the lines of maxima converge towards all singular points within the interval, without being limited to any particular ones. Then, the relative maximum of the wavelet modulus satisfies the following requirements at chunk $x_m$ with scale $s$:
		\begin{equation}
			\begin{aligned}
				\! \left| W_{f}(x_m, {s\!+\!1})\right| \! < \! \left| W_{f}(x_m, s)\right| \! > \! \left| W_{f}(x_m, {s\!-\!1})\right| \!, s \! \subseteq \! Z \label{WTMM}.
			\end{aligned}
		\end{equation}
		The initial number of segments $K_f$ and the value of $\Theta^s$ for candidate highlight segmentations of video $f$ are obtained by these chunks through the wavelet transform modulus maxima method. Then, (P0) can be equivalently transformed into
		\begin{alignat}{2}
			\vspace{-1mm}
			\label{P1}
			(\rm{P1}): \quad & \begin{array}{*{20}{c}}
				\mathop {\max }\limits_{\theta_{x,k,f} \in \Theta}E\left(\theta_{x,k,f}\right)
			\end{array} & \\ 
			\mbox{s.t.}\quad
			& \left( \ref{P0}a\right), \left( \ref{P0}b\right), \left( \ref{P0}f\right), \nonumber \\
			& \theta_{x,k,f} = \theta_{x,f}^c\theta_{x,k}^s, \forall x\in {\mathcal{X}}_f, k \in{ \mathcal{K}}_f, f \in\mathcal{F}, & \tag{\ref{P1}a}\\
			& \sum_{f=1}^{F}\theta_{x,k,f} \leq 1, \forall x\in{\mathcal{X}}_f, k \in {\mathcal{K}}_f,f \in\mathcal{F}, & \tag{\ref{P1}b} \\
			& \theta_{x,k,f} = \theta_{x',k,f}, \forall {x, x'} \in {\mathcal{X}}_f, k \in {\mathcal{K}}_f, f \in\mathcal{F}. & \tag{\ref{P1}c}
			\vspace{-3mm}
		\end{alignat}
		
		\subsection{Highlight-Direction Trimming Algorithm}
		In this subsection, a quick local optimal algorithm is presented to solve the highlight trimming iteratively. Algorithm \ref{alg2} illustrates the explicit descriptions of the highlight-direction trimming for adaptive vehicle travel on road. \par
		\begin{prop}\label{direction}
			Given the selected chunk $x_{m}^{t-1}$ of video $f$ at $(t-1)$th iteration, a local optimal video chunk selection at
			the iteration is given by
			\begin{equation}
				x^{t} \in\left\lbrace x_{m}^{t-1}-1, x_{m}^{t-1}+1, x_{m}^{t}\right\rbrace \label{nextsel},
			\end{equation}  
			where $x_{m}^{t}$ is the $t$th iteration highlight segmentation chunk obtained according to Lemma \ref{MM} and equation (\ref{WTMM}) and $x_{m}^{t}\leq x_{m}^{t-1}$ holds for all $x_{m}^{t}, x_{m}^{t-1}\in {\mathcal{X}}_f$.
		\end{prop}
		
		\begin{proof}
			If the highlight entropy of selecting adjacent chunks $\left\lbrace x_{m}^{t-1}-1, x_{m}^{t-1}+1\right\rbrace $ satisfies
			\begin{equation}
				\begin{split}
					& \dfrac{(y_f(x_{m}^{t-1})+y_f(x_{m}^{t-1}\pm 1))^2}{\sum_{k=1}^{K_f}\sum_{x=1}^{X_f}y_f^2(x)-y_f^2(x_{m}^{t-1})-y_f^2(x_{m}^{t-1}\pm 1)}\\
					& \geq
					\dfrac{y_f^2(x_{m}^{t-1})+y_f^2(x_{m}^{t})}{\sum_{k=1}^{K_f}\sum_{x=1}^{X_f}y_f^2(x)-y_f^2(x_{m}^{t-1})-y_f^2(x_{m}^{t})}, 
				\end{split}
			\end{equation}
			then selecting adjacent chunks to ensure the segment inner continuous according to the restriction (\ref{P1}c) will obtain maximal video highlight entropy at step $t$, otherwise if the first greater than or equal to symbol is false, a chunk $x_{m}^{t}$ with higher video chunks’ popularity value according to Lemma \ref{MM} will be selected to attain maximal highlight entropy.
			$\hfill\blacksquare$
		\end{proof}
		Proposition \ref{direction} implies the selection of chunk caching for the next iteration directed towards obtaining optimal video chunks for highlight trimming. The use of candidate highlight segmentation based on CWT significantly reduces the traversal number in each iteration to optimize video file pre-caching.\par
		The complexity of the highlight-direction trimming algorithm is $O(N_f\left( K_fn\right) )$, in which $N_f, K_f, n$ is the number of videos, chunks, and iterations. The iteration converges quickly with only a few iterations since it starts from the peak value of the video chunks’ popularity.
		\vspace{-3mm}
		\begin{algorithm}[H]
			\caption{Highlight-direction Trimming Algorithm}\label{alg2}
			\begin{algorithmic}[1]
				\Require
				$R, \mathcal{C}, \text{video chunks’ popularity } y(x), \text{vehicle speed } v$
				\State \textbf{Initial:} $y\left( x_{m}\right)$ by CWT according to Eq.(\ref{WT}) and Eq.(\ref{WTMM})
				\While{$\theta_{x,k,f}$ meet the constraints}
				\State Update the next selection set $x^{t}$ by Eq.(\ref{nextsel})
				\State Caculate $\sum_{f=1}^{F}E_f$ according to Eq.(\ref{HE})
				\State Update $\theta_{x,k,f}$ by solving problem (P1)
				\EndWhile
			\end{algorithmic}
		\end{algorithm}
		\vspace{-6mm}
		\section{Simulation Results and Analysis}
		To verify the functionality and feasibility of our proposed pre-caching scheme, the subjective experiment is conducted to analyze QoE, subjective jitter, and highlight level. The performance of the proposed pre-caching scheme will be verified through three criteria in the objective experiment,
		including the highlight entropy, objective jitter, and cache hit ratio. These experiments are based on video datasets in the real world. The baseline algorithms for comparison are briefly described as follows respectively.
		\begin{itemize}[leftmargin=*]
			\item \textbf{\textit{N-Speed Playback (NSP)}} \cite{bNSP}: Skips $N-1$ chunks every time viewing a chunk. This mechanism, a basic feature of most video players, is designed for fast-forwarding users.
			\item \textbf{\textit{Selecting segments with popularity As High As Possible (AHAP)}}: Selects highlights by greedy algorithm for tending-to-skip users \cite{bAHAP}.
			\item \textbf{\textit{Elitism-Based Compact Genetic Algorithms (EGA)}} \cite{b37}: Making sure that the best individuals are not discarded, by transferring them directly into the next generation.
		\end{itemize}
		\subsection{Simulation Setup}
		To analyze the effectiveness of our proposed caching scheme and its impact on QoE, jitter, and highlight level, a week-long experiment is conducted, which involves many participants who were asked to watch nine video clips and rate the clips. The video clips are obtained from three different caching schemes with three different playback durations. \par
		In the objective simulation, a 20 km long freeway with 4 RSUs is considered. The coverage range of each RSU is 5 km. The video chunks’ popularity comes from video frames' popularity according to user browsing behaviors' datasets of video providers such as YouTube, iQIYI, etc \cite{b44}. According to the works of \cite{b1} and taking into consideration the specific context of our work, we set the chunk duration $\tau$ to $10*30$ with a video frame rate of $30 fps$. The parameters for simulations are given in Table \ref{tab1}. \par
		\begin{table}
			\begin{center}
				\vspace{-2mm}
				\caption{Parameters and Values}
				\vspace{-2mm}
				\label{tab1}
				\begin{tabular}{ c  c }
					\hline
					\text{Parameter} & \text{Value} \\
					\hline
					$F, M, \tau$ & $10, 4, 10*30$ \\
					$T_d$ & $30, 45, 60 \text{ min}$\\
					$r_m, z_f, d_m, \mathcal{C}$ & $1.8 \text{ Mbps}, 2.25 \text{ MB}, 5 \text{ km}, 64 \text{ Gb}$ \\
					$v_m$ & $\left\lbrace 120, 60, 40, 30, 24, 20, 17, 15\right\rbrace  \text{ km/h}$ \\
					\hline
				\end{tabular}
				\vspace{-8mm}
			\end{center}
		\end{table}
		
		\subsection{Simulation Results}
		In Fig.~{\ref{subfig}}, the subjective video viewing performance are illustrated respectively under different playback durations (denoted by $T_d$) for all the pre-caching schemes. The QoE is denoted by the average rating according to volunteers. The subjective jitter represents the incoherence of lines, actions, background music, etc. The highlight level indicates whether all the popular segments of a video that volunteers want to watch are included. Specifically, two major insights can be seen in Fig.~{\ref{figure7}}. First, watching video segments obtained using our strategy yields high QoE scores for all playback durations. Then, when the playback duration is short, our strategy for watching video segments achieves much higher QoE compared to others. The main reason is that the candidate highlight segmentation is optimized by highlight entropy considering not only video segments’ popularity but also the continuity of segments within a period of time for viewing experience improvement. As the playback duration approaches the full length of the video, the advantage becomes less prominent, as all the segments have been watched and the highlight level tends to saturate, while the jitter tends to zero. Fig.~{\ref{figure6}} demonstrates that our strategy effectively reduces the jitter of viewing video segments, with a decrease in jitter as the playback duration increases, since the number and duration of playout segments increase. Fig.~{\ref{figure4}} shows that the viewing highlight level of the proposed scheme is no less than that of benchmark schemes as playback duration increases since video highlights are watched.\par
		\begin{figure*}[ht]
			\centering
			\vspace{-2mm}
			\setlength{\abovecaptionskip}{0.cm}
			\subfigure[Subjective QoEs under different playback durations.]
			{	
				\label{figure7}
				\includegraphics[width=5.6cm, height=3.8cm]{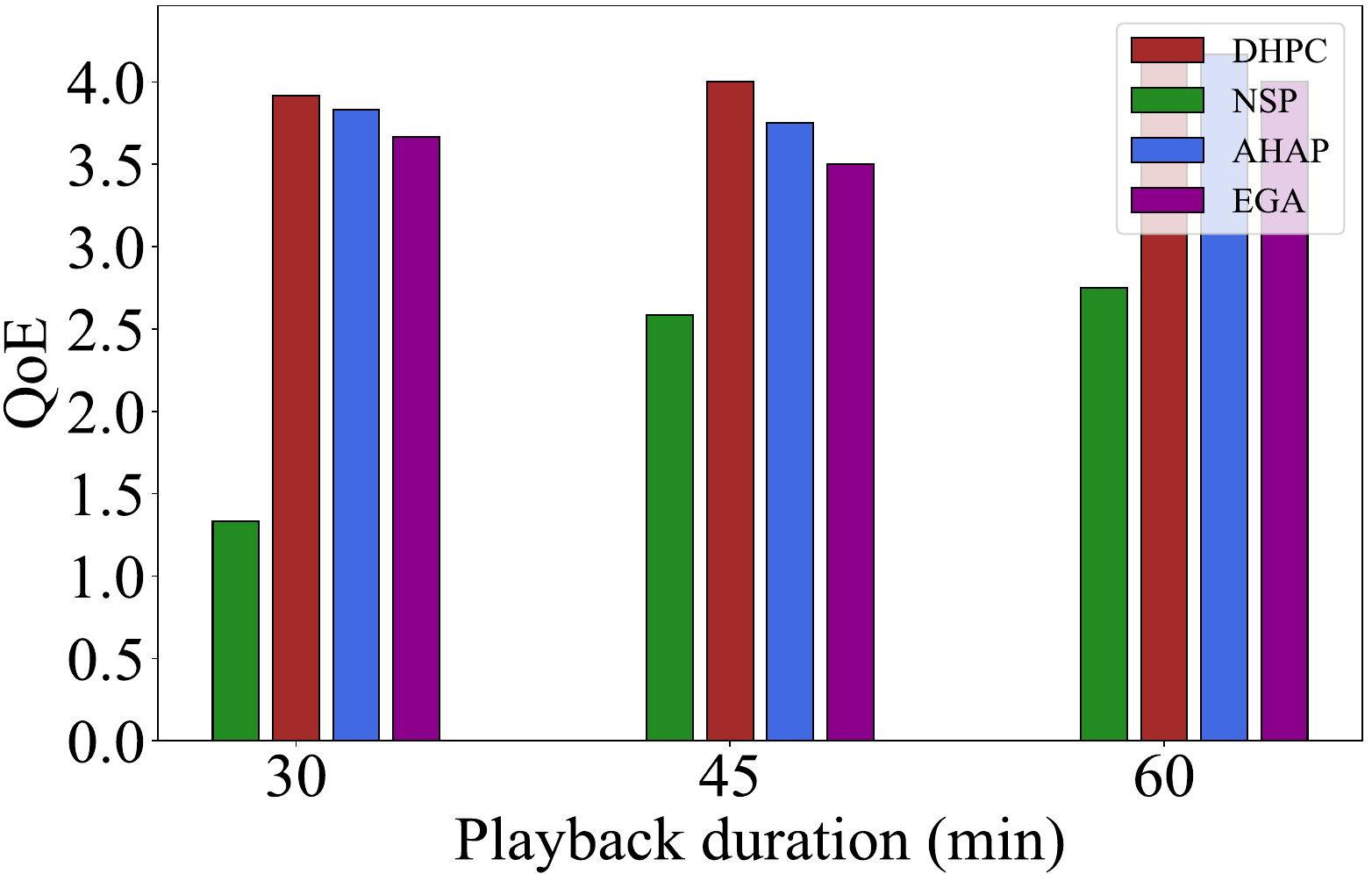}
			}
			\subfigure[Subjective jitter degrees under different playback durations.]
			{	
				\label{figure6}
				\includegraphics[width=5.6cm, height=3.8cm]{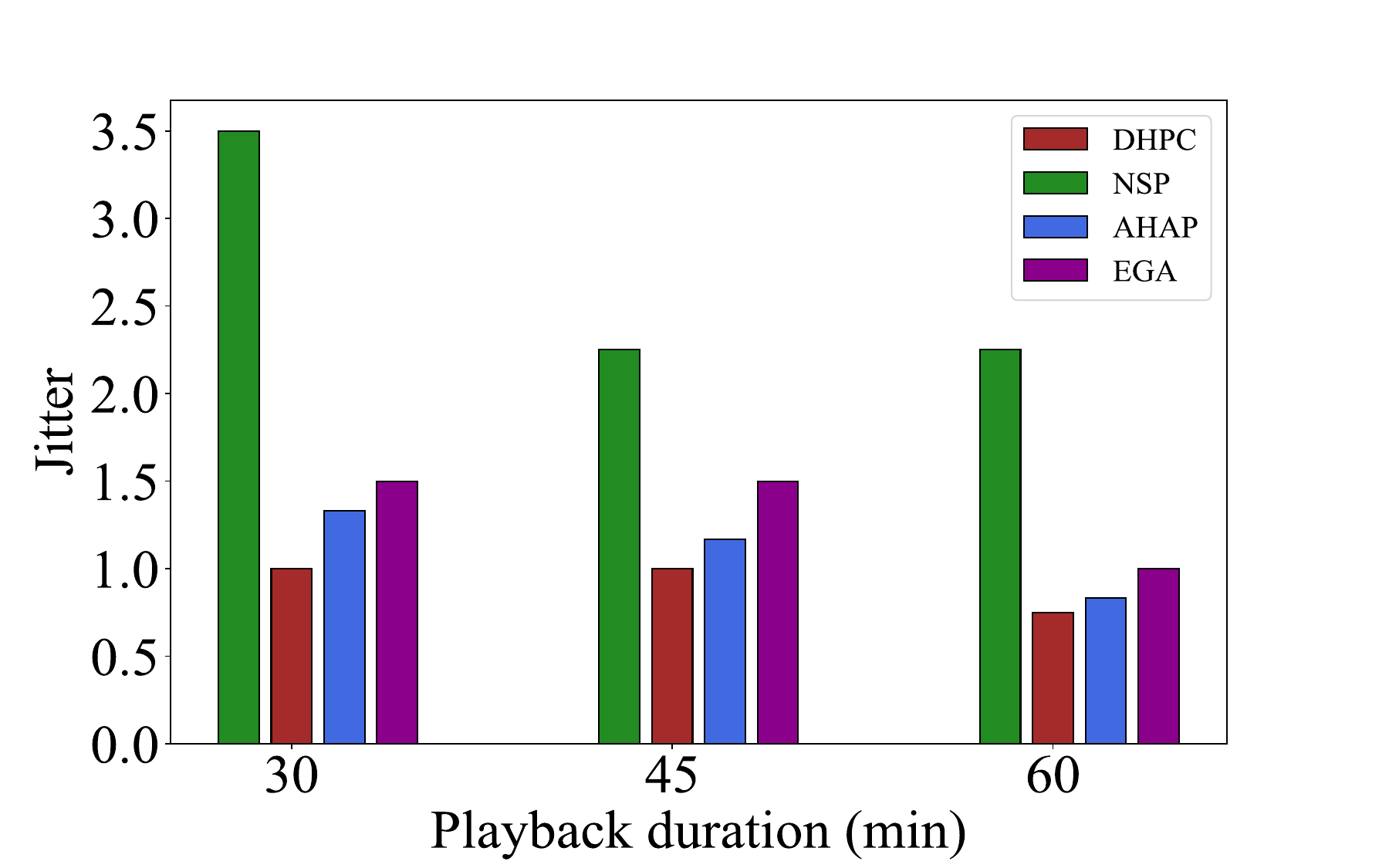}
			}
			\subfigure[Subjective highlight levels under different playback durations.]
			{	
				\label{figure4}
				\includegraphics[width=5.6cm, height=3.8cm]{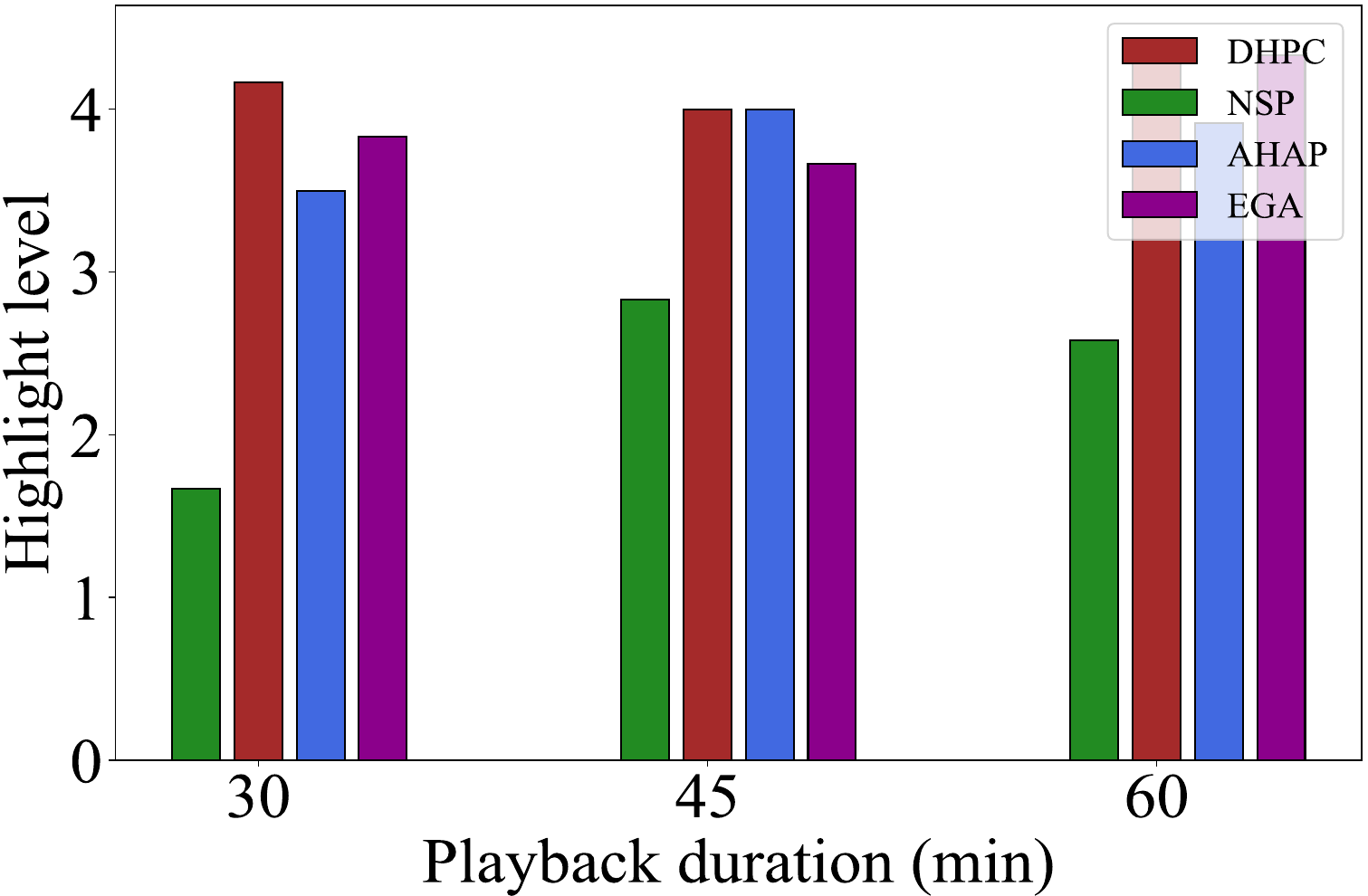}
			}	
			\vspace{-1mm}
			\caption{Achievable subjective evaluation parameters.}
			\vspace{-2mm}
			\label{subfig}
		\end{figure*}
		\begin{figure*}[ht]
			\centering
			\vspace{-3mm}
			\setlength{\abovecaptionskip}{0.cm}
			\subfigure[HEs under different service durations.]
			{
				\label{fig8}
				\includegraphics[width=5.6cm, height=3.8cm]{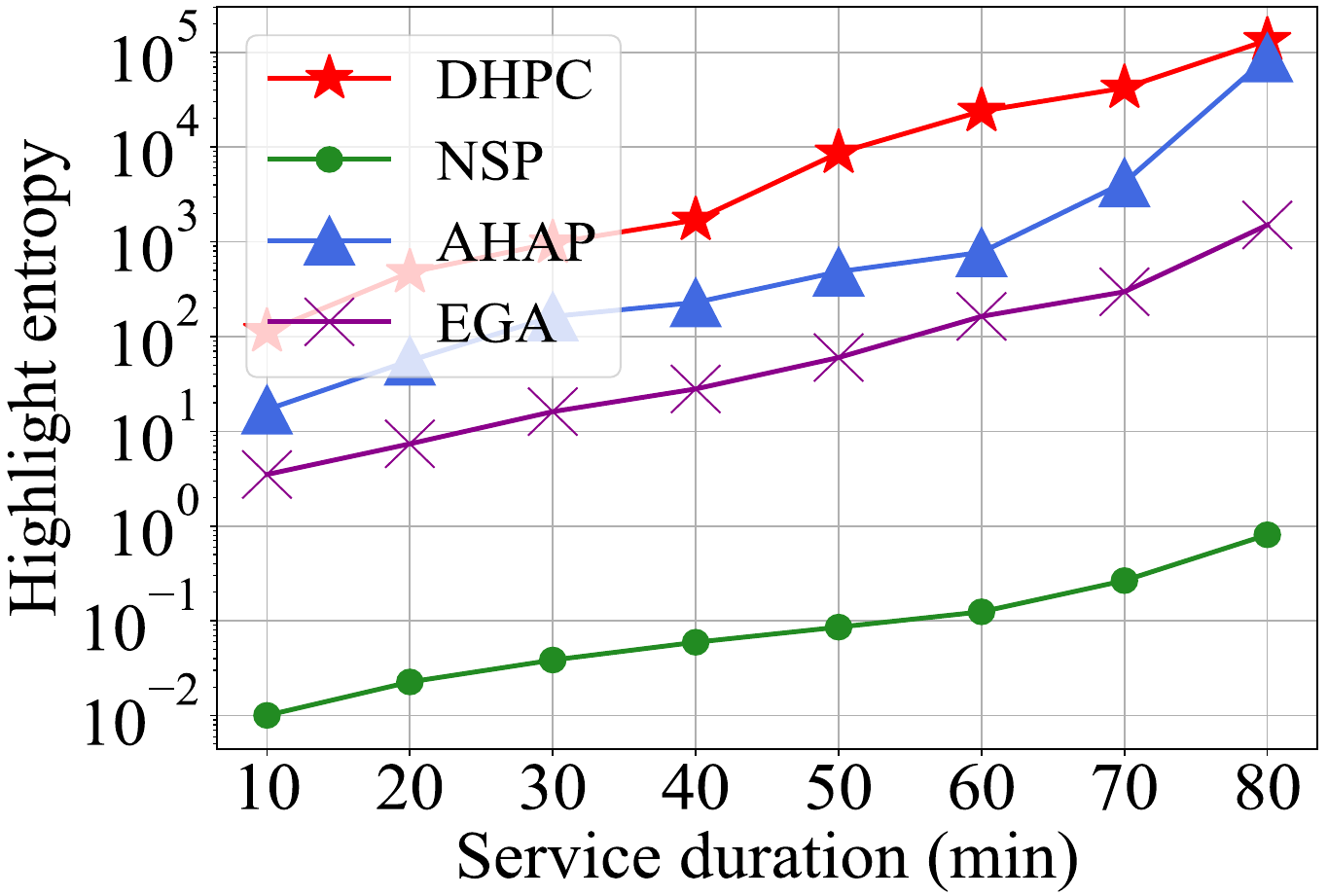}
			}
			\subfigure[Jitters under different service durations.]
			{
				\label{fig9}
				\includegraphics[width=5.6cm, height=3.8cm]{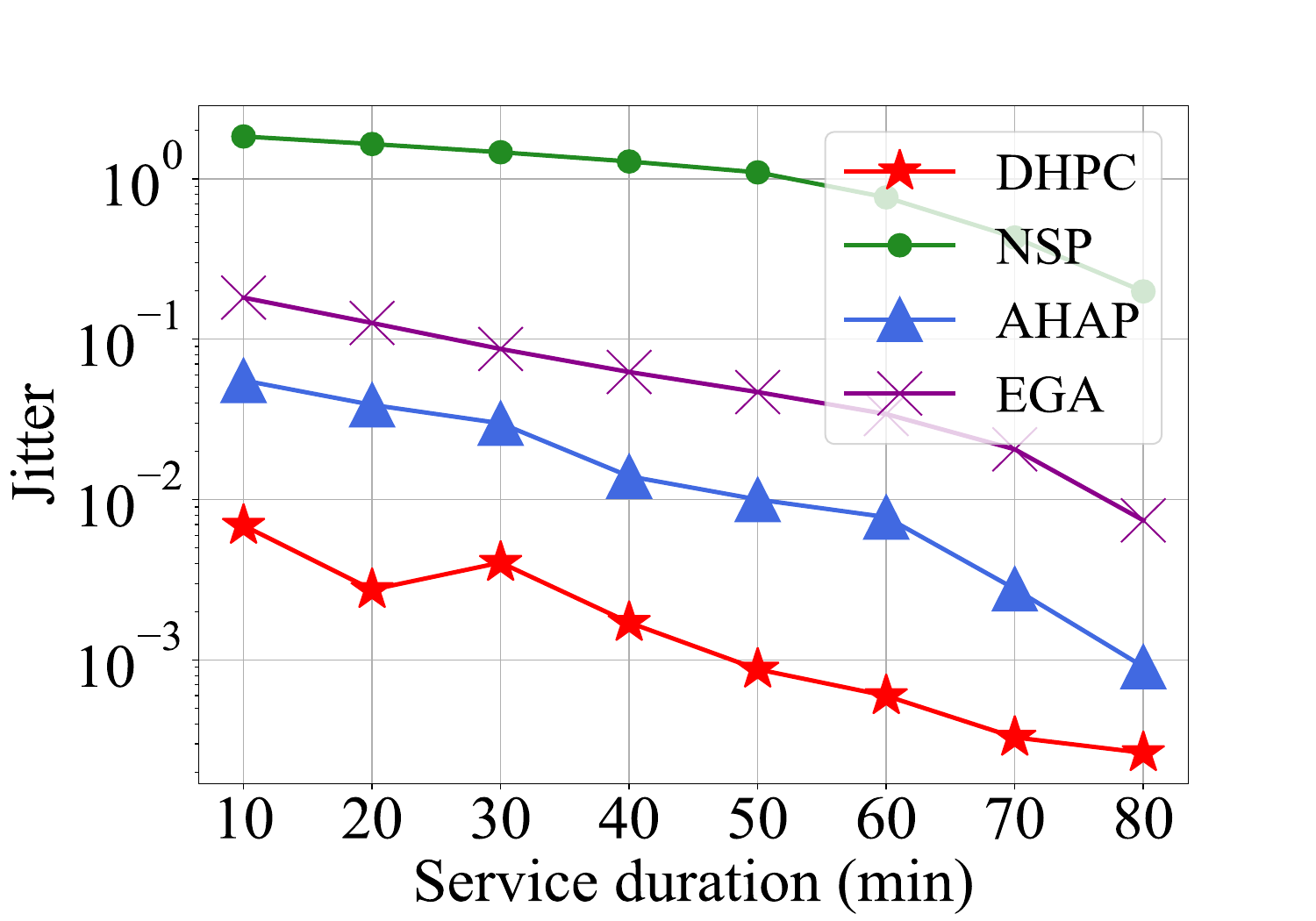}
			}
			\subfigure[Cache hit ratios under different viewing ratios.]
			{
				\label{fig10}
				\includegraphics[width=5.6cm, height=3.8cm]{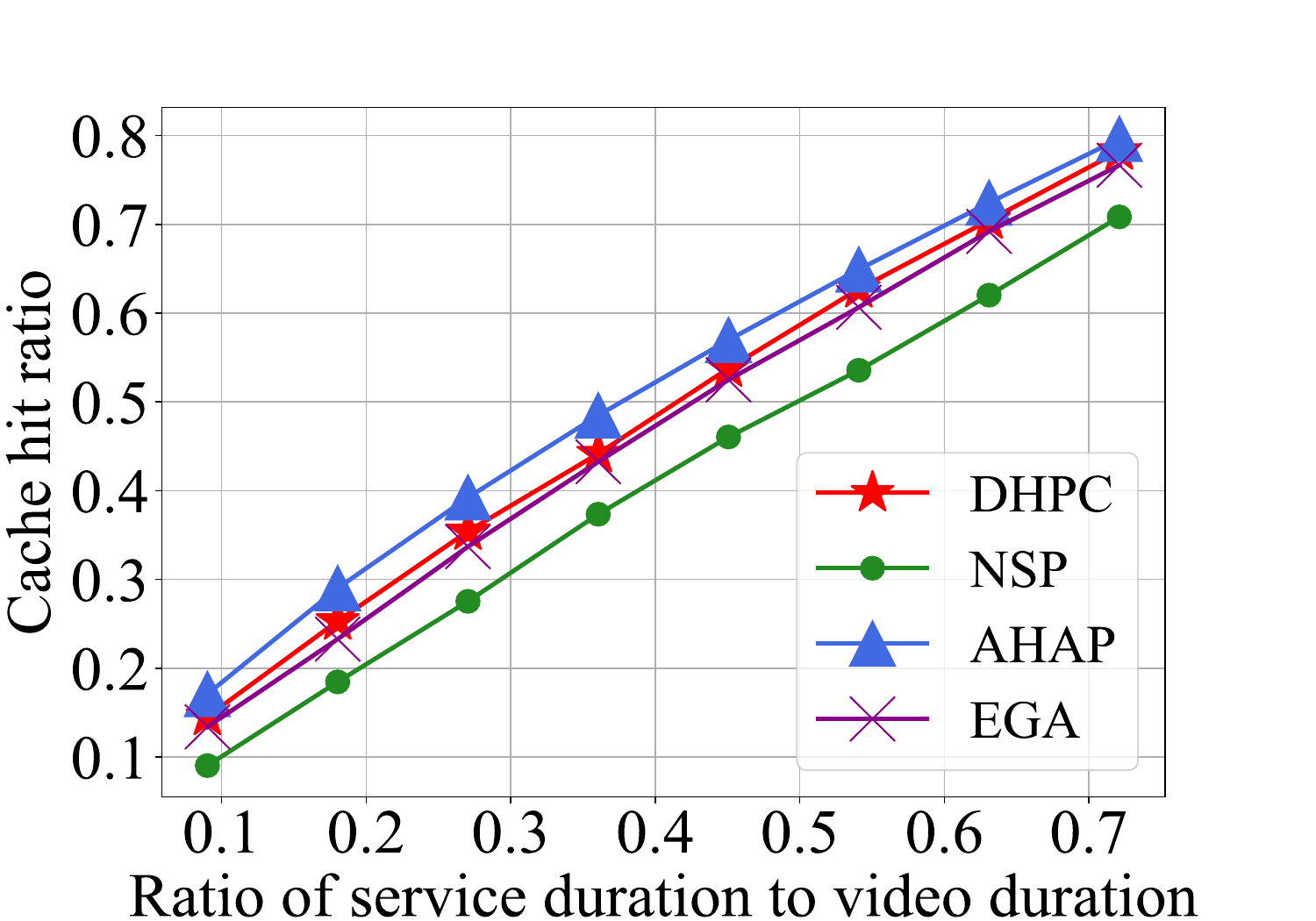}
			}
			\vspace{-1mm}
			\caption{Evaluation parameters of the different pre-caching strategies.}
			\vspace{-6mm}
			\label{obf}
		\end{figure*}
		The objective viewing video performances obtained by different strategies at different service durations are shown in Fig.~{\ref{obf}}. The objective jitter is the ratio of the average interval chunk number between adjacent viewing segments to the total chunk number of a video. The cache hit ratio is the sum of watched segments' popularity. As illustrated in Fig.~{\ref{fig8}}, our proposed method can achieve higher highlight entropy for short service durations. This is because that DHPC pre-caches more consecutive segments with higher video chunks’ popularity. However, as all popular chunks have been selected, there have minimal differences in the highlight entropy between AHAP and DHPC. In Fig.~{\ref{fig9}}, it can be observed a significant reduction in overall jitter as the service duration increases. This is because our strategy aims to reduce the jitter between highlights and enhance continuity during a highlight. In general, our proposed DHPC pre-caches highlight segments with great improvement in the quality of users’ viewing experience. Fig.~{\ref{fig10}} illustrates that our proposed pre-caching scheme outperforms the EGA and NSP methods and drops behind the AHAP in terms of cache hit ratio across different ratios of service duration to video duration. The main reason is that to improve overall viewing experience quality, the proposed scheme, EGA, and NSP fail to maintain higher video segments' popularity while ensuring the continuity of the highlight segments.
		\section{Conclusion and Future Works}
		To enable mobility vehicles to view video with high quality on road, we proposed a duration-adaptive highlight video pre-caching scheme in this paper. To balance the popularity and continuity between segments within a period of time, we constructed a highlight entropy based quality evaluation model and formulated a highlight entropy maximization problem with the limitation of resources and service duration constraints. Based on wavelet transform we proposed a highlight-direction trimming algorithm. Theoretical analysis indicated that the proposed algorithm can achieve local optimum. Subsequently, the algorithm is evaluated on a real dataset, and the simulation results verify that our proposed algorithm improves the video viewing highlight entropy and jitter. 
		For future work, we will focus on the design of dynamic pre-caching with joint optimization of viewing highlight segments and global balance to enhance the system's adaptivity.
		
		\bibliography{references}

\begin{thebibliography}{10}
\providecommand{\url}[1]{#1}
\csname url@samestyle\endcsname
\providecommand{\newblock}{\relax}
\providecommand{\bibinfo}[2]{#2}
\providecommand{\BIBentrySTDinterwordspacing}{\spaceskip=0pt\relax}
\providecommand{\BIBentryALTinterwordstretchfactor}{4}
\providecommand{\BIBentryALTinterwordspacing}{\spaceskip=\fontdimen2\font plus
\BIBentryALTinterwordstretchfactor\fontdimen3\font minus
  \fontdimen4\font\relax}
\providecommand{\BIBforeignlanguage}[2]{{%
\expandafter\ifx\csname l@#1\endcsname\relax
\typeout{** WARNING: IEEEtran.bst: No hyphenation pattern has been}%
\typeout{** loaded for the language `#1'. Using the pattern for}%
\typeout{** the default language instead.}%
\else
\language=\csname l@#1\endcsname
\fi
#2}}
\providecommand{\BIBdecl}{\relax}
\BIBdecl

\bibitem{bDT}
Z.~Hu, S.~Lou, Y.~Xing, X.~Wang, D.~Cao, and C.~Lv, ``Review and perspectives
  on driver digital twin and its enabling technologies for intelligent
  vehicles,'' \emph{IEEE Trans. Veh. Technol.}, vol.~7, no.~3, pp. 417--440,
  Aug. 2022.

\bibitem{b40}
J.~Chen, H.~Wu, P.~Yang, F.~Lyu, and X.~Shen, ``Cooperative edge caching with
  location-based and popular contents for vehicular networks,'' \emph{IEEE
  Trans. Veh. Technol.}, vol.~69, no.~9, pp. 10\,291--10\,305, Sept. 2020.

\bibitem{b10}
B.~Jedari, G.~Premsankar, G.~Illahi, M.~D. Francesco, A.~Mehrabi, and
  A.~Ylä-Jääski, ``Video caching, analytics, and delivery at the wireless
  edge: A survey and future directions,'' \emph{IEEE Commun. Surveys Tuts.},
  vol.~23, no.~1, pp. 431--471, 1th Quart. 2021.

\bibitem{b46}
X.~Jiang, F.~R. Yu, T.~Song, and V.~C.~M. Leung, ``Resource allocation of video
  streaming over vehicular networks: A survey, some research issues and
  challenges,'' \emph{IEEE Trans. Intell. Transp. Syst.}, vol.~23, no.~7, pp.
  5955--5975, Jul. 2022.

\bibitem{b1}
K.~Li, Z.~Lyu, H.~Liu, and P.~Fan, ``A popularity-and mobility-aware
  multi-layer caching with feedback mechanism for highway vehicular networks,''
  in \emph{2020 IEEE 92nd Vehicular Technology Conference
  (VTC2020-Fall)}.\hskip 1em plus 0.5em minus 0.4em\relax Victoria, BC, Canada,
  Nov. 2020.

\bibitem{b34}
Z.~Zhao, L.~Guardalben, M.~Karimzadeh, J.~Silva, T.~Braun, and S.~Sargento,
  ``Mobility prediction-assisted over-the-top edge prefetching for hierarchical
  vanets,'' \emph{IEEE J. Sel. Areas Commun.}, vol.~36, no.~8, pp. 1786--1801,
  Aug. 2018.

\bibitem{b41}
A.~Mahmood, C.~E. Casetti, C.~F. Chiasserini, P.~Giaccone, and J.~Härri, ``The
  rich prefetching in edge caches for in-order delivery to connected cars,''
  \emph{IEEE Trans. Veh. Technol.}, vol.~68, no.~1, pp. 4--18, Jan. 2019.

\bibitem{bperchunkcaching}
L.~Vigneri, T.~Spyropoulos, and C.~Barakat, ``Low cost video streaming through
  mobile edge caching: Modelling and optimization,'' \emph{IEEE Trans.
  Comput.}, vol.~18, no.~6, pp. 1302--1315, Jun. 2019.

\bibitem{chunkcharacterize}
L.~Maggi, L.~Gkatzikis, G.~Paschos, and J.~Leguay, ``Adapting caching to
  audience retention rate,'' \emph{Computer Communications}, vol. 116, pp.
  159--171, 2018.

\bibitem{meng2022sensing}
K.~Meng, Q.~Wu, W.~Chen, and D.~Li, ``Sensing-assisted communication in
  vehicular networks with intelligent surface,'' \emph{arXiv preprint
  arXiv:2211.11475}, Dec. 2022.

\bibitem{VCG}
H.~Andrew, ``Youtube rolls out activity graph to all videos, ups the maximum
  price of channel memberships,'' May 2022.

\bibitem{b25}
Q.~Xu, Z.~Su, and Q.~Yang, ``Blockchain-based trustworthy edge caching scheme
  for mobile cyber-physical system,'' \emph{IEEE Internet Things J.}, vol.~7,
  no.~2, pp. 1098--1110, Feb. 2020.

\bibitem{b26}
C.-L. Tu, W.-L. Hwang, and J.~Ho, ``Analysis of singularities from modulus
  maxima of complex wavelets,'' \emph{IEEE Trans. Inf. Theory}, vol.~51, no.~3,
  pp. 1049--1062, Mar. 2005.

\bibitem{bNSP}
\BIBentryALTinterwordspacing
Zeekless, ``What does it mean for media to play at n speed?'' Mar. 2021.
  [Online]. Available:
  \url{https://math.stackexchange.com/questions/4050571/what-does-it-mean-for-media-to-play-at-1-5x-or-nx-speed}
\BIBentrySTDinterwordspacing

\bibitem{bAHAP}
R.~Schaback and H.~Wendland, ``Adaptive greedy techniques for approximate
  solution of large rbf systems,'' \emph{Numerical Algorithms}, vol.~24, no.~3,
  pp. 239--254, 2000.

\bibitem{b37}
C.~W.~A. et~al., ``Elitism-based compact genetic algorithms,'' \emph{IEEE
  Trans. Evol. Comput.}, vol.~7, no.~4, pp. 367--385, Aug. 2003.

\bibitem{b44}
\BIBentryALTinterwordspacing
L.~Tianwei, ``How to accurately identify "highlights" in variety show videos,''
  May 2017. [Online]. Available:
  \url{https://cloud.tencent.com/developer/article/1166117}
\BIBentrySTDinterwordspacing

\end{thebibliography}
		\bibliographystyle{IEEEtran}
		
\end{document}